\UseRawInputEncoding
\documentclass[conference,letterpaper]{IEEEtran}

\IEEEoverridecommandlockouts
\usepackage[T1]{fontenc}
\usepackage{url}
\usepackage{ifthen}
\usepackage[cmex10]{amsmath}
\usepackage{bbm}
\usepackage{amsfonts}
\usepackage[dvips]{graphicx}
\usepackage{times}
\usepackage{cite}
\usepackage{amsmath}
\usepackage{array}
\usepackage{amssymb}

\usepackage{stfloats}
\usepackage{slashbox}
\usepackage{graphicx}
\usepackage{subfigure}
\usepackage{footnote}
\usepackage{booktabs}
\usepackage{array}
\usepackage{algorithm}
\usepackage{subeqnarray}
\usepackage{cases}
\usepackage{threeparttable}
\usepackage{color}
\usepackage{epstopdf}
\usepackage{algpseudocode}
\usepackage{bm}
\usepackage{multirow}
\usepackage{adjustbox}

\newtheorem{theorem}{Theorem}
\newtheorem{remark}{Remark}

\newtheorem{corollary}{Corollary}

\begin{document}

\title{Sum-Rate-Distortion Function for Indirect Multiterminal Source Coding in Federated Learning}

\author{%
  \IEEEauthorblockN{Naifu Zhang, Meixia Tao and Jia Wang}\thanks{This work is supported by the NSF of China under grant 61941106 and the NSF of China under grant 61771305.}
  \IEEEauthorblockA{Department of Electronic Engineering\\
                    Shanghai Jiao Tong University, Shanghai, 200240, P. R. China\\
                    Emails: \{arthaslery, mxtao, jiawang\}@sjtu.edu.cn}
}

\maketitle

\begin{abstract}
One of the main focus in federated learning (FL) is the communication efficiency since a large number of participating edge devices send their updates to the edge server at each round of the model training.
Existing works reconstruct each model update from edge devices and implicitly assume that the local model updates are independent over edge devices.
In FL, however, the model update is an indirect multi-terminal source coding problem, also called as the CEO problem where each edge device cannot observe directly the gradient that is to be reconstructed at the decoder, but is rather provided only with a noisy version.
The existing works do not leverage the redundancy in the information transmitted by different edges.
This paper studies the rate region for the indirect multiterminal source coding problem in FL.
The goal is to obtain the minimum achievable rate at a particular upper bound of gradient variance.
We obtain the rate region for the quadratic vector Gaussian CEO problem under unbiased estimator and derive an explicit formula of the sum-rate-distortion function in the special case where gradient are identical over edge device and dimension.
Finally, we analyse communication efficiency of convex Mini-batched SGD and non-convex Minibatched SGD based on the sum-rate-distortion function,
respectively.
\end{abstract}

\textit{A full version of this paper is accessible at:}
\url{https://github.com/arthaslery/My-Paper/blob/main/rate-distortion.pdf}

\section{Introduction}
Federated learning (FL) \cite{mcmahan2016communication,konen2016federated,bonawitz2019towards,yang2019federated} is a new edge learning framework that enables many edge devices to collaboratively train a machine learning model without exchanging datasets under the coordination of an edge server.
In FL, each edge device downloads a shared model from the edge server, computes an update to the current model by learning from its local dataset, then sends this update to the edge server.
Therein, the updates are averaged to improve the shared model.
Compared with traditional learning at a centralized data center, FL offers several distinct advantages, such as preserving privacy, reducing
network congestion, and leveraging distributed on-device computation.
FL has recently attracted significant attention from both academia and industry, such as \cite{9026922,9003425,9084352,tao2020zte}.

The main focus in the research area is communication-efficient FL.
Specifically, the communication efficient FL is to achieve better convergence rate (high model accuracy) with lower communication costs.
The works of SGD convergence analysis \cite{MAL-050,doi:10.1137/120880811} state that the convergence rate mainly depends on the variance bound of gradients and the number of updates.
The communication cost depends on the communication cost per update and the number of updates.
Several recent methods have been proposed to improve communication-efficiency in federated settings, including aggregation frequency control \cite{mcmahan2016communication,201564,8664630}, compression schemes \cite{seide20141,alistarh2018qsgd,zhou2018dorefanet,NIPS2017_89fcd07f,10.1007/978-3-319-46493-0_32,pmlr-v80-bernstein18a,pmlr-v54-leblond17a,aji2017sparse,lin2018deep,tsuzuku2018variance} and user scheduling \cite{8761315,9107235,amiri2020convergence,9053740}.

All the existing researches aim to reconstruct each model update from edge devices and implicitly assume that the local model updates are independent over edge devices.
However, we observe that the main objective in FL is a good estimate of a model update at the edge server by using the information received from edge devices, rather than the exact recovery of each model update from each devices, which is an indirect multi-terminal source coding problem.
Specifically, in FL, the objective is to estimate the global model update computed by gradient decent (GD) on global dataset, while the local model update computed by each edge device is a noisy version of the global model update.
In addition, the local model updates are highly correlated among different edge devices, providing opportunity for correlated source coding.
Hence, the existing works do not leverage the redundancy in the information transmitted by different edge device.

Motivated by the above issue, in this paper, we derive the rate region for the indirect multiterminal source coding problem in FL.
Our goal is to obtain the minimum achievable rate at a particular upper bound of gradient variance.
We formulate the indirect multiterminal source coding problem in FL and solve it from the standpoint of multiterminal rate-distortion theory.
Our result can be regarded as a tool to analyse the communication efficiency in a certain FL system.
The main contributions of this work are outlined below:
\begin{itemize}
\item\emph{Rate region results:}
We reveal that the multiterminal source coding problem in FL is the quadratic vector Gaussian CEO problem base on a thorough understanding of gradient distributions.
We derive the rate region for the quadratic vector Gaussian CEO problem under unbiased estimator.
For the achievability proof, we adopt the classic Berger-Tung scheme \cite{10016434852} but design an unbiased estimator at receiver.
Our converse proof is inspired by Oohama's converse in work \cite{669162} but tightens the converse bound in work \cite{669162} by the application of unbiased estimator.

\item\emph{Communication efficiency analysis:}
We derive a closed-form sum-rate-distortion function in the special case where gradient are identical over edge device and dimension.
We analyse communication efficiency of convex Minibatched SGD and non-convex Minibatched SGD based on the sum-rate-distortion function, respectively.
We provide an inherent trade-off between communication cost and convergence guarantees.
\end{itemize}

\section{Federated Learning}
To facilitate the presentation, we only focus on a basic FL setup.
However, the results can be extended to cases with gradient distribution under our assumptions.
In this section, we introduce the system model and the convergence rate of federated learning in error-free communication.

\subsection{System Model}
\begin{figure}[t]
\begin{centering}
\vspace{-0.2cm}
\includegraphics[scale=.40]{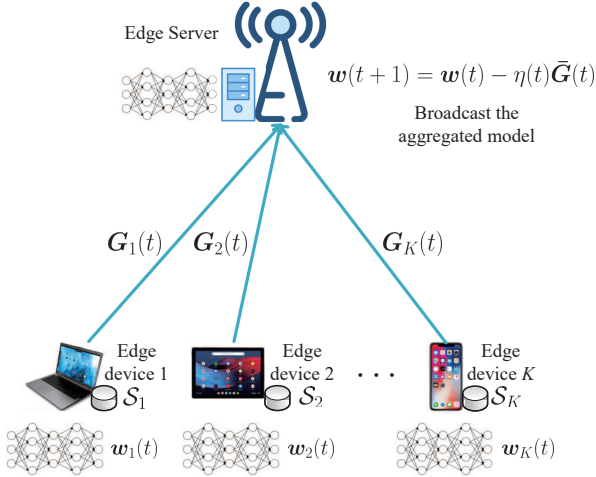}
\vspace{-0.1cm}
 \caption{\small{Illustration of federated learning system in error-free communication.}}\label{fig:system}
\end{centering}
\vspace{-0.3cm}
\end{figure}
We consider a FL framework as illustrated in Fig.~\ref{fig:system}, where a shared AI model (e.g., a classifier) is trained collaboratively across $K$ edge devices via the coordination of an edge server.
Let $\mathcal{K}=\{1,...,K\}$ denote the set of edge devices.
Each device $k\in\mathcal{K}$ collects a fraction of labelled training data via interaction with its own users, constituting a local dataset, denoted as $\mathcal{S}_k$.
Let $\bm{w}\in\mathbb{R}^P$ denote the $P$-dimensional model parameter to be learned.
The loss function measuring the model error is defined as
\begin{equation}
F(\bm{w})=\sum_{k\in\mathcal{K}}{\frac{|\mathcal{S}_k|}{|\mathcal{S}|}F_k(\bm{w})},
\label{equ:loss_function}
\end{equation}
where $F_k(\bm{w})=\frac{1}{|\mathcal{S}_k|}\sum_{i\in\mathcal{S}_k}{f_i(\bm{w})}$ is the loss function of device $k$ quantifying the prediction error of the model $\bm{w}$ on the local dataset collected at the $k$-th device, with $f_i(\bm{w})$ being the sample-wise loss function, and $\mathcal{S}=\bigcup_{k\in\mathcal{K}}{\mathcal{S}_k}$ is the union of all datasets.
Let $\bm{G}(t)\triangleq\nabla F(\bm{w}(t))\in\mathbb{R}^P$ denote the gradient vector calculated through the gradient descent (GD) algorithm at iteration $t$.
The minimization of $F(\bm{w})$ is typically carried out through the Minibatched stochastic gradient descent (Minibatched SGD) algorithm, where device $k$'s local dataset $\mathcal{S}_k$ is split into mini-batches of size $B_k$ and at each iteration $t=1,2,...$, we draw one mini-batch $\mathcal{B}_k(t)$ randomly and calculate the local gradient vector as
\begin{equation}
\bm{G}_k(t)=\nabla\frac{1}{B_k}\sum_{i\in\mathcal{B}_k(t)}{f_i(\bm{w})}.
\end{equation}
When the mini-batch size $B_k=1$, the Minibatched SGD algorithm reduces to SGD algorithm.
In this case, we say the local gradient $\bm{G}_k(t)$ has variance $\sigma^2(t)$ at iteration $t$, i.e., $\mathbb{E}[\|\bm{G}_k(t)-\bm{G}(t)\|^2]=\sigma^2(t)$.
In general case, the local gradient $\bm{G}_k(t)$ has variance $\frac{\sigma^2(t)}{B_k(t)}$ at iteration $t$.
If all the local gradients $\{\bm{G}_k(t)\}_{k=1}^K$ are available at edge server through error-free transmission, the optimal estimator of $\bm{G}(t)$ is the Sample Mean Estimator, i.e., $\bm{\bar{G}}(t)=\sum_{k=1}^K{\frac{B_k(t)}{B(t)}\bm{G}_k(t)}$, where $B(t)\triangleq\sum_{k=1}^K{B_k(t)}$ is global batch size.
It is not hard to see that the variance of the optimal estimator $\bm{\bar{G}}(t)$ is $\frac{\sigma^2(t)}{B(t)}$.
Then the edge server update the model parameter as
\begin{equation}
\bm{w}(t+1)=\bm{w}(t)-\eta(t)\bm{\bar{G}}(t),
\label{equ:centralized_update}
\end{equation}
with $\eta(t)$ being the learning rate at iteration $t$.

\subsection{Convergence rate}
Given access to local gradients, and a starting point $\bm{w}(0)$, Minibatched SGD builds iterates $\bm{w}(t)$ given by Equation (\ref{equ:centralized_update}), projected onto $\mathbb{R}^P$, where $\{\eta\}_{t\geq 0}$ is a sequence of learning rate.
In this setting, one can show:
\begin{theorem}[\cite{MAL-050}, Theorem 6.3]
\label{theorem:convergence_rate}
Let $F:\mathbb{R}^P\rightarrow\mathbb{R}$ be unknown, convex and $L$-smooth. Let $\bm{w}(0)$ be given, and let $A^2=\sup_{\bm{w}\in\mathbb{R}^P}{\|\bm{w}-\bm{w}(0)\|^2}$.
Let $T>0$ be fixed.
Given repeated, independent access to the unbiased estimator of gradients with variance bound $D$ for loss function $F$, i.e., $\mathbb{E}[\|\bm{\bar{G}}(t)-\bm{G}(t)\|^2]\leq D$ for all $t=1,2,...,T$, training with initial point $\bm{w}(0)$ and constant step sizes $\eta(t)=\frac{\gamma}{L+1}$, where $\gamma=A\sqrt{\frac{2}{DT}}$, achieves

\begin{align}
\mathbb{E}\left[L\left(\frac{1}{T}\sum_{t=1}^T{\bm{w}(t)}\right)\right]-\min_{\bm{w}\in\mathbb{R}^P}{L(\bm{w})}\leq A\sqrt{\frac{2D}{T}}+\frac{LA^2}{T}.\label{equ:convergence_rate}
\end{align}
\end{theorem}
\begin{remark}
When the gradient vector calculated through minibatched SGD algorithm, and the model is updated with the unbiased estimator $\{\bm{\bar{G}}(t)\}_{t=1}^T$ in error-free transmission, the variance bound $D=\max\limits_{t=1,2,...,T}\frac{\sigma^2(t)}{B(t)}$.
However, the error-free transmission is infeasible in practice due to communication resource limitations.
The local gradients have to be quantized and the unbiased estimation can only be based on these quantized values.
Hence, the variance of the unbiased estimator based on quantized gradients must be larger than $\frac{\sigma^2(t)}{B(t)}$ at each iteration $t$.
\end{remark}

In general, the convergence bound (\ref{equ:convergence_rate}) increases with the variance bound of gradient estimator while the required communication bits per iteration decrease with the variance bound of gradient estimator.
To obtain the trade-off between the convergence rate and the communication cost, we need to find out what is the minimum achievable rate at a particular variance upper bound.
It is a basic problem in rate distortion theory.

\section{Multi-terminal source coding problem in FL}
In this section, in order to accurately formulate the indirect multiterminal source coding problem in federated learning, we first study the distribution of the global gradients and the local gradients in FL.
All the assumptions of the gradient distributions are justified by experiment on datasets such as MNIST in the full version of this paper.
Then we formulate a quadratic vector Gaussian CEO problem based on a thorough understanding of gradient distributions.

\subsection{Gradient distribution}
\subsubsection{Global Gradient}
Recall that the edge server is interested in this sequence $\{\bm{G}(t)\}_{t=1}^\infty$, which is the global gradient vector sequence calculated through the gradient descent (GD) algorithm.
The following are key assumptions on the distribution of $\{\bm{G}(t)\}_{t=1}^\infty$
\begin{itemize}
\item The gradient $\{\bm{G}(t)\}_{t=1}^\infty$ is independent distributed vector Gaussian sequence with zero mean. The Gaussian distribution is valid since the field of probabilistic modeling uses Gaussian distributions to model the gradient \cite{sra2012optimization,ida2013domain,6917199,miyashita2013nonparametric}. The assumption of independence is valid when the learning rate is large enough.

\item The gradient $\bm{G}(t)$ are independent over gradient vector dimension $p$'s. This assumption is valid as long as the features in a data sample are independent. Even if the gradients are strongly correlated over dimension $p$'s, the gradients can be de-correlated by the regularization methods such as sparsity-inducing regularization \cite{collins2014memory,scardapane2017group} and parameter sharing/tying \cite{dieleman2016exploiting,chen2015compressing}.
\end{itemize}

\subsubsection{Local Gradients}
For $k=1,2,...,K$, edge device $k$ carries out Minibatched SGD in FL.
Recall that $\{\bm{G}_k(t)\}_{t=1}^\infty$ is the local gradient vector sequence calculated through the Minibatched SGD algorithm at device $k$.
The local gradient vector sequence $\{\bm{G}_k(t)\}_{t=1}^\infty$ can be viewed as noisy version of $\{\bm{G}(t)\}_{t=1}^\infty$ and corrupted by additive noise, i.e., $\bm{G}_k(t)=\bm{G}(t)+\bm{N}_k(t)$.
The following are key assumptions on the distribution of $\{\bm{N}_k(t)\}_{t=1}^\infty$ for $k\in\mathcal{K}$:
\begin{itemize}
\item The gradient noise $\bm{N}_k(t)$ are Gaussian random vectors independent of the $\bm{G}$ process. The mean of $\bm{N}_k(t)$ is zero in IID data setting and non-zero in non-IID data setting. Note that the gradient noise $\bm{N}_k(t)$ depends on the selection of local batch at edge device $k$. The assumption of independence is valid since the selection of local batch is independent of the global gradient $\bm{G}$.

\item The gradient noise $\bm{N}_k(t)$ are independent and non-identical distributed over devices $k$'s. This assumption is valid as long as the selection of the local batch are independent and non-identical over edge device $k$.

\item The gradient noise $\bm{N}_k(t)$ are independent and non-identical distributed over dimension $p$'s. The reason for this assumption is similar to that of gradient $\bm{G}(t)$.
\end{itemize}

\subsection{Problem Formulation}
In this subsection, we formulate the quadratic vector gaussian CEO problem based on the observation of the distribution of the global gradient $\bm{G}(t)$ and the local gradient $\bm{G}_k(t)$.
Let the global gradient $\{\bm{G}(t)\}_{t=1}^\infty$ be an independent Gaussian vector sequence with mean 0 and variance $\mbox{diag}(\sigma_{X_1}^2(t),\sigma_{X_2}^2(t),...,\sigma_{X_P}^2(t))$.
Each $\bm{G}(t),t=1,2,...$ takes value in real space $\mathbb{R}^P$.
For $k=1,...,K$, Let the local gradient $\bm{G}_k(t)$ be noisy version of $\{\bm{G}(t)\}_{t=1}^\infty$, each taking value in real space $\mathbb{R}^P$ and corrupted by independent additive white Gaussian noise, i.e.,
\begin{align}
\bm{G}_k(t)=\bm{G}(t)+\bm{N}_k(t),
\end{align}
where $\bm{N}_k(t)$ are Gaussian random vectors independent over device $k$, dimension $p$ and iteration $t$.
For $k=1,2,...,K$, $p=1,2,...,P$ and $t=1,2,...$, we assume that $\bm{N}_k(t)$ is a centralized Gaussian variable with mean 0 and variance
$\mbox{diag}(\sigma_{N_{k,1}}^2(t),\sigma_{N_{k,2}}^2(t),...,\sigma_{N_{k,P}}^2(t))$.

\begin{figure}[t]
\begin{centering}
\vspace{-0.2cm}
\includegraphics[scale=.35]{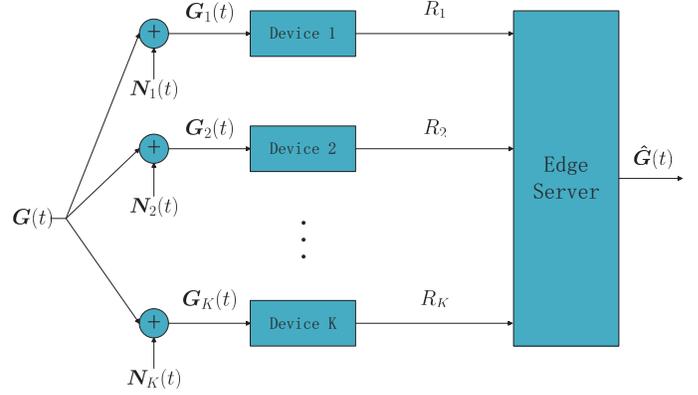}
\vspace{-0.1cm}
 \caption{\small{The CEO problem in FL.}}\label{fig:CEO}
\end{centering}
\vspace{-0.3cm}
\end{figure}
Fig.~\ref{fig:CEO} shows the CEO Problem in FL.
The edge server (CEO) is interested in the sequence $\{\bm{G}(t)\}_{t=1}^\infty$ that cannot be observed directly.
The edge server employs a team of $K$ edge devices (agents) who observes independently corrupted versions $\{\bm{G}_k(t)\}_{t=1}^\infty,k=1,2,...,K$ of $\{\bm{G}(t)\}_{t=1}^\infty$.
We write $n$ independent copies of $\bm{G}(t)$ and $\bm{G}_k(t)$ as $\bm{G}^n(t)$ and $\bm{G}_k^n(t)$, respectively.
To facilitate the following derivation, we omit iteration index $t$.
For $k=1,2,...,K$, each local gradient sequence $\bm{G}_k^n$ observed by edge device $k$ is separately encoded to $\phi_k(\bm{G}_k^n)$, and those are sent to the information processing center, where the edge server observes $\phi_k(\bm{G}_k^n),k=1,2,...,K$ and outputs the estimation $\bm{\hat{g}}^n$ of $\bm{G}^n$ by using the decoder function $\psi_K$.
The encoder function $\phi_k,k=1,2,...,K$ are defined by
\begin{align}
\phi_k:\mathbb{R}^{nP}\rightarrow\mathcal{C}_k=\{1,2,...,|\mathcal{C}_k|\},
\end{align}
and satisfy the total rate constraint
\begin{align}
\frac{1}{n}\log{|\mathcal{C}_k|}\leq R_k,k=1,2,...,K.
\label{equ:rate_constrain}
\end{align}
We write a $K$-tuple of encoder functions $\phi_k,k=1,2,...,K$ as
\begin{align}
\phi^K=(\phi_1,\phi_2,...,\phi_K).
\end{align}
Similarly, we write
\begin{align}
\phi^K(\bm{G}^{nK})=(\phi_1(\bm{G}_1^n),\phi_2(\bm{G}_2^n),...,\phi_K(\bm{G}_K^n)).
\end{align}
The decoder function $\psi_K$ is defined by
\begin{align}
\psi_K:\mathcal{C}_1\times\mathcal{C}_2\times...\times\mathcal{C}_K\rightarrow\mathbb{R}^{nP}.
\end{align}

For $\bm{\hat{G}}^n=\psi_K(\phi^K(\bm{G}^{nK}))$, define the average mean squared error (MSE) distortion by
\begin{align}
D^n(\bm{G}^n,\bm{\hat{G}}^n)=\frac{1}{n}\sum_{i=1}^n{\mathbb{E}\|\bm{G}[i]-\bm{\hat{G}}[i]\|^2}.
\end{align}
For a target distortion $D$, a rate $K$-tuple $(R_1,R_2,...,R_K)$ is said to be achievable if there are encoders $\phi^K$ satisfying (\ref{equ:rate_constrain}) and decoder $\psi_K$ such that $\bm{\hat{G}}^n$ is unbiased estimator of $\bm{G}^n$, i.e., $\mathbb{E}[\bm{\hat{G}}^n|\bm{G}^n=\bm{g}^n]=\bm{g}^n$ and $D^n(\bm{G}^n,\bm{\hat{G}}^n)\leq D$ for some $n$.
The closure of the set of all achievable rate $K$-tuples is called the rate-region and we denote it by $\mathcal{R}_\star\subseteq\mathbb{R}_+^K$.
Our aim is to characterize the region $\mathcal{R}_\star$ in an explicit form.

\begin{remark}
This problem is different from the existing Quadratic Gaussian CEO works \cite{669162,1365154,5508637}.
The estimator of our studied problem is unbiased while the existing works do not have the unbiased constraint.
We assume that the estimator must be unbiased because Theorem \ref{theorem:convergence_rate} requires that the gradient estimator is unbiased.
In addition, a sequence of biased gradients with bounded distortion cannot guarantee the convergence of federated learning.
\end{remark}

\section{Rate Region Results}
Let $\mathcal{R}_N(D)$ be the Berger-Tung achievable region using Gaussian auxiliary random variables for CEO problem in FL.
We show that
\begin{align}
&\mathcal{R}_N(D)=\\
\bigg\{&(R_1,...,R_K)\in\mathbb{R}_+^K:R_k\geq\sum_{p=1}^P{R_{k,p}},k=1,2,...,K\\
&\forall (R_{1,p},...,R_{K,p})\in\mathcal{R}_p(r_{1,p},...,r_{K,p}),p=1,2,...,P\\
&\forall (r_{1,p},...,r_{K,p})\in\mathcal{F}_p(D_p),p=1,2,...,P\label{region:r_kp}\\
&\forall (D_1,D_2,...,D_P)\in\mathbb{R}_+^P,\sum_{p=1}^P{D_p}\le D\bigg\}\label{region:D_p}
\end{align}
where
\begin{align}
&\mathcal{R}_{p}(r_{1,p},r_{2,p},...,r_{K,p})=\bigg\{(R_{1,p},R_{2,p},...,R_{K,p}):\nonumber\\
&\sum_{k\in\mathcal{A}}{R_{k,p}}\geq\sum_{k\in\mathcal{A}}{r_{k,p}}+\frac{1}{2}\log{(\frac{1}{\sigma_{X_p}^2}+\frac{1}{D_p})}\\
&-\frac{1}{2}\log{\left(\frac{1}{\sigma_{X_p}^2}+\sum_{k\in\mathcal{A}^c}{\frac{1-\exp(-2r_{k,p})}{\sigma_{N_{k,p}}^2}}\right)},\forall \mathcal{A}\subseteq\mathcal{K}\bigg\},
\label{region:set_R_p}
\end{align}
and
\begin{align}
\mathcal{F}_p(D_p)=\left\{(r_{1,p},...,r_{K,p}):\sum_{k=1}^K{\frac{1-\exp(-2r_{k,p})}{\sigma_{N_{k,p}}^2}}=\frac{1}{D_p}\right\}.
\end{align}

Our main result is
\begin{theorem}
\label{theorem:rate_region}
\begin{align}
\mathcal{R}_\star(D)=\mathcal{R}_N(D).
\end{align}
\end{theorem}

\begin{remark}
Parameter $D_p$ can be interpreted as the distortion of gradient estimator at dimension $p$, $r_{k,p}$ can be interpreted as the rate of the $k$-th edge device spends in quantizing its observation noise at dimension $p$, and $R_{k,p}$ can be interpreted as the rate contributed by edge device $k$ on dimension $p$.
It can be observed that our rate region with unbiased estimation constraint is a subset of the classic rate region without such constraint.
In the special case when the distortion $D$ is larger than the gradient variance $\sum_{p=1}^P{\sigma_{X_p}^2}$, the classic rate region can achieve zero.
This indicates that we can simply set $\bm{\hat{G}}=0$ at the receiver without any transmission.
However, $\bm{\hat{G}}=0$ is not an option in our unbiased setting since the corresponding estimator is biased.
Thus in this large distortion case, the rate region under unbiased estimator is still bounded away from zero.
The extra rate introduced by our rate region results is necessary for model training.
For a distortion $D$ larger than the gradient variance $\sum_{p=1}^P{\sigma_{X_p}^2}$, the model will never converge by applying the classic rate region while can converge by applying the rate region under unbiased estimator.
\end{remark}

We prove this theorem from the standpoint of multiterminal rate-distortion theory.
For the achievability proof, we adopt the classic Berger-Tung scheme \cite{10016434852} but design an unbiased estimator at receiver.
Our converse proof is inspired by Oohama's converse in \cite{669162} but is not straightforward and cannot be derived directly from traditional (biased) CEO.
We tighten the term $\frac{1}{2}\log(\frac{\sigma_X^2}{D})$ at (8) in work \cite{669162} to $\frac{1}{2}\log(1+\frac{\sigma_X^2}{D})$ by the application of unbiased estimator.
Due to page limit, the proof of Theorem \ref{theorem:rate_region} is omitted.
It can be found in the full version of this paper.

\section{Communication Efficiency of Federated Learning} 
In this section, we provide an inherent trade-off between communication cost and convergence guarantees based on the sum-rate-distortion function.
Our sum-rate-distortion function is quite portable, and can be applied to almost any stochastic gradient method.
For illustration, we analyse communication efficiency of convex Minibatched SGD and non-convex Minibatched SGD, respectively.

\subsection{Sum-rate-distortion function}
The rate region in section IV is implicit.
To facilitate analyse the communication efficiency of FL, we will first derive the sum-rate-distortion function in this section.

The sum-rate-distortion function is defined by
\begin{align}
R_{sum}(D)\triangleq\min_{(R_1,R_2,...,R_K)\in\mathcal{R}_\star(D)}{\sum_{k\in\mathcal{K}}{R_k}}.
\end{align}
Recall that the rate region $\mathcal{R}_\star(D)$ in (\ref{region:D_p}), we have
\begin{align}
\min_{(R_1,R_2,...,R_K)\in\mathcal{R}_\star(D)}\sum_{k\in\mathcal{K}}R_k=\min_{D_p}\min_{r_{k,p}}\min_{R_{k,p}}{\sum_{k\in\mathcal{K}}\sum_{p=1}^P{R_{k,p}}}.
\end{align}
Note that $\sum_{k\in\mathcal{K}}R_{k,p}\geq\sum_{k\in\mathcal{K}}r_{k,p}+\frac{1}{2}\log(1+\frac{\sigma_{X_p}^2}{D_p})$, therefore we can obtain the following sum-rate-distortion function by solving an optimization problem.
\begin{corollary}
For every $D>0$, when the gradient elements are identical over dimension $p$'s and local gradients are identical over device $k$'s.
\begin{align}
R_{sum}(D)=P\left[-\frac{K}{2}\log\left(1-\frac{\sigma_{N}^2}{KD}\right)+\frac{1}{2}\log(1+\frac{\sigma_{X}^2}{D})\right],
\end{align}
where $\sigma_{X}^2$ is the variance of global gradient and $\sigma_{N}^2$ is the variance of gradient noise.
\label{corollary:idential_sun_rate_distortion_function}
\end{corollary}
\begin{remark}
The derived function has the form of a sum of two nonnegative functions.
The first term decreases with the number of edge devices and is dominant for relatively small $D$.
The second term is a classical channel capacity for a Gaussian channel with power constraint $\sigma_{X}^2$ and noise variance $D$.
\end{remark}

\subsection{Convex Minibatched SGD}
Combining the sum-rate-distortion function given in Corollary \ref{corollary:idential_sun_rate_distortion_function} and the guarantees for Minibatched SGD algorithms on smooth, convex functions yield the following results:
\begin{theorem}[Smooth Convex Optimization]
\label{theorem:convex_communication_efficiency}
Let $F,L,\bm{w}(0)$ and $A$ be as in Theorem 1.
Fix $\epsilon>0$.
Let $\sigma_{X}^2(t)$ be the variance of global gradient and $\sigma_{N}^2(t)$ be the variance of gradient noise at iteration $t$.
Suppose the edge server outputs the unbiased gradient estimate $\{\bm{\hat{G}}(t)\}_{t=1}^T$ from $K$ identical edge devices accessing independent stochastic gradients with variance bound $D$, i.e., $\mathbb{E}\|\bm{G}(t)-\bm{\hat{G}}(t)\|^2\leq D$ for all $t=1,2,...,T$, and with step size $\eta(t)=\frac{\gamma}{L+1}$, where $\gamma$ is as in Theorem 1.

To guarantee the convergence rate $\mathbb{E}\left[L\left(\frac{1}{T}\sum_{t=1}^T{\bm{w}(t)}\right)\right]-\min_{\bm{w}\in\mathbb{R}^P}{L(\bm{w})}\leq\epsilon$, the number of iterations should satisfy
\begin{align}
T\geq A^2\left(\sqrt{\frac{D}{2\epsilon^2}+\frac{L}{\epsilon}}+\sqrt{\frac{D}{2\epsilon^2}}\right)^2=O\left(A^2\max(\frac{2D}{\epsilon^2},\frac{L}{\epsilon})\right)
\label{equ:number_of_iterations}
\end{align}
Moreover, the required communication bits at iteration $t$ are given by
\begin{align}
P\left[-\frac{K}{2}\log\left(1-\frac{\sigma_{N}^2(t)}{KD}\right)+\frac{1}{2}\log(1+\frac{\sigma_{X}^2(t)}{D})\right].
\label{equ:lower_bound_communication}
\end{align}
\end{theorem}
\begin{proof}
To guarantee the convergence rate $\mathbb{E}\left[L\left(\frac{1}{T}\sum_{t=1}^T{\bm{w}(t)}\right)\right]-\min_{\bm{w}\in\mathbb{R}^P}{L(\bm{w})}\leq\epsilon$, from Theorem 1 we should have
\begin{align}
A\sqrt{\frac{2D}{T}}+\frac{LA^2}{T}\leq\epsilon\Leftrightarrow T\geq A^2\left(\sqrt{\frac{D}{2\epsilon^2}+\frac{L}{\epsilon}}+\sqrt{\frac{D}{2\epsilon^2}}\right)^2.
\end{align}

By the definition of rate region, there exists encoders $\phi^K(t)$ and decoder $\psi_K(t)$ satisfying that the sum rate achieves $R_{sum}(D)(t)=P\left[-\frac{K}{2}\log\left(1-\frac{\sigma_{N}^2(t)}{KD}\right)+\frac{1}{2}\log(1+\frac{\sigma_{X}^2(t)}{D})\right]$ such that $\mathbb{E}[\bm{\hat{G}}^n(t)|\bm{G}^n(t)=\bm{g}^n(t)]=\bm{g}^n(t)$ and $D^n(\bm{G}^n(t),\bm{\hat{G}}^n(t))\leq D$ for a large enough sequence length $n$.
Therefore, $R_{sum}(D)(t)$ is the required communication bits at iteration $t$.
\end{proof}
\begin{remark}
In the most reasonable regimes, the first term of the max in (\ref{equ:number_of_iterations}) will dominate the number of iterations necessary.
Specifically, the number of iterations will depend linearly on the estimate gradient variance bound $D$.
In SGD-based learning, the gradient variance $\sigma_{X}^2(t)$ is large at the begining, then gradually approaches to zero when the model converges.
The noise variance $\sigma_{N}^2(t)$, on the other hand, remains approximately unchanged due to the randomness of local datasets throughout the training process.
As a result, according to (\ref{equ:lower_bound_communication}), the required communication bits per iteration decrease during the training and converges to $-P\frac{K}{2}\log\left(1-\frac{\sigma_{N}^2(t)}{KD}\right)$.
\end{remark}

\subsection{Non-convex Minibatched SGD}
In many interesting applications such as neural network training, however, the objective is non-convex, where much less is known.
However, there has been an interesting line of recent work which shows that Minibatched SGD at least always provably converages to a local minima, when $F$ is smooth.
For instance, by applying Theorem 2.1 in \cite{doi:10.1137/120880811}, we immediately obtain the following communication efficiency results for FL.
\begin{theorem}[Smooth Non-convex Optimization]
Let $F:\mathbb{R}^P\rightarrow\mathbb{R}$ be a $L$-smooth (possibly non-convex) function, and Let $\bm{w}(0)$ be given.
Let the iteration limit $T>0$ be fixed and the random stoping iteration with probability mass function supported on $\{1,...,T\}$.
Suppose the edge server outputs the unbiased gradient estimate $\{\bm{\hat{G}}(t)\}_{t=1}^T$ from $K$ identical edge devices accessing independent stochastic gradients with variance bound $D$, i.e., $\mathbb{E}\|\bm{G}(t)-\bm{\hat{G}}(t)\|^2\leq D$ for all $t=1,2,...,T$, and with step size $\eta=O(\frac{1}{L})$, then we have
\begin{align}
\frac{1}{L}\mathbb{E}[\|\nabla F(\bm{w})\|^2]\leq O\left(\frac{F(\bm{w}(0))-F^*}{N}+\frac{D}{L}\right).
\end{align}
Moreover, the required communication bits at iteration $t$ are the same as in Theorem \ref{theorem:convex_communication_efficiency}.
\end{theorem}
\section{Conclusion}
This paper studied the indirect multiterminal source coding problem in FL.
We have characterized the rate region and show that our rate region with unbiased estimation constraint is a subset of the classic rate
region without such constraint.
We also derived the closed-form sum-rate-distortion function in special case.
Finally, we apply the sum-rate-distortion function to analyse the communication efficiency of FL.
We provide an inherent trade-off between communication cost and convergence guarantees based on the sum-rate-distortion function.
In the future work, we will study the practical distributed source coding scheme based on the estimated gradient distribution to achieve the rate-distortion function in FL.
We will also study the communication-efficient user scheduling and rate allocation schemes based on the explicit rate region in FL.

\bibliographystyle{IEEEtran}
\bibliography{IEEEabrv,rate-distortion}
\end{document}